\documentclass{IEEEtran}

\usepackage{xcolor}
\usepackage{cite}
\usepackage{url}
\usepackage{lipsum}
\usepackage{graphicx}
\usepackage{amsfonts}
\usepackage{amsmath}
\usepackage[ruled,noline,linesnumbered,noend]{algorithm2e}
\usepackage{booktabs}
\usepackage[bookmarks=false]{hyperref}
\usepackage[USenglish]{babel} 
\usepackage{authblk}
\usepackage{breqn}
\usepackage{amsthm}

\usepackage{array}
\usepackage{graphbox}
\usepackage{subfigure}
\usepackage{tikz}
\usepackage{float}

\DeclareMathOperator*{\argmin}{arg\,min}
\newcommand{\norm}[1]{\left\lVert#1\right\rVert}

\title{
CodEx: A Modular Framework for Joint Temporal De-blurring and Tomographic Reconstruction
}

\author{
Soumendu~Majee,
Selin~Aslan,
Do\u ga~G\" ursoy~\IEEEmembership{Member,~IEEE}
Charles~A.~Bouman~\IEEEmembership{Fellow,~IEEE},

\thanks{Soumendu Majee is with Samsung Research America (e-mail: \href{mailto:soumendu.majee.1@gmail.com}{soumendu.majee.1@gmail.com}).
Selin Aslan is with the X-ray Science Division, Argonne National Laboratory (e-mail: \href{mailto:selinaslanphd@gmail.com}{selinaslanphd@gmail.com}).
Doga Gursoy is with the X-ray Science Division, Argonne National Laboratory and the Electrical Engineering and Computer Science Department, Northwestern University  (e-mail: \href{mailto:dgursoy@anl.gov}{dgursoy@anl.gov}).
Charles A. Bouman is with the School of Electrical and Computer Engineering, Purdue University (e-mail: \href{mailto:bouman@purdue.edu}{bouman@purdue.edu}).

This research used resources of the Advanced Photon Source, a U.S. Department of Energy (DOE) Office of Science User Facility operated for the DOE Office of Science by Argonne National Laboratory under Contract No. DE-AC02-06CH11357.
This work was partially supported by an NSF grant number CCF-1763896.
}
}

\begin{document}


\IEEEtitleabstractindextext{%
\begin{abstract}

In many computed tomography (CT) imaging applications, it is important to rapidly collect data from an object that is moving or changing with time. 
Tomographic acquisition is generally assumed to be step-and-shoot, where the object is rotated to each desired angle, and a view is taken.
However, step-and-shoot acquisition is slow and can waste photons, so in practice fly-scanning is done where the object is continuously rotated while collecting data.
However, this can result in motion-blurred views and consequently reconstructions with severe motion artifacts.

In this paper, we introduce CodEx, a modular framework for joint de-blurring and tomographic reconstruction that can effectively invert the motion blur introduced in sparse view fly-scanning.
The method is a synergistic combination of a novel acquisition method with a novel non-convex Bayesian reconstruction algorithm. 
CodEx works by encoding the acquisition with a known binary code that the reconstruction algorithm then inverts.
Using a well chosen binary code to encode the measurements can improve the accuracy of the inversion process.
The CodEx reconstruction method uses the alternating direction method of multipliers (ADMM) to split the inverse problem into iterative deblurring and reconstruction sub-problems, making reconstruction practical to implement.
We present reconstruction results on both simulated and binned experimental data to demonstrate the effectiveness of our method.

\end{abstract}

}

\maketitle
\IEEEdisplaynontitleabstractindextext

\section{Introduction}

Computed tomography (CT) imaging has been widely used in a variety of applications to study the internal structure of static and dynamic objects.

Traditionally, the spatial and temporal resolution of CT reconstruction has been limited by effects due to interpolation~\cite{crawford1990computed}, scintillator decay-rates~\cite{rutherford2016evaluating}, and X-ray bow-tie filters~\cite{harpen1999simple} among others.
More recently, temporal resolution has also become important for the imaging of moving objects~\cite{flohr2001heart}.

For rapidly changing objects, good temporal resolution is crucial in order to resolve the reconstructed object accurately.
Model based iterative reconstruction (MBIR)~\cite{bouman1996unified} techniques have led to significant improvements in temporal resolution for time-resolved CT through the use of novel view-sampling~\cite{mohan2015timbir} and improved prior modeling~\cite{majee2021multi,majee20194d,mohan2015timbir,majee2021high} .
However, even with these improvements, the temporal resolution is fundamentally limited by the rate of collection of projection measurements.

\begin{figure}[ht]
\centering     
\includegraphics[width=0.5\textwidth]{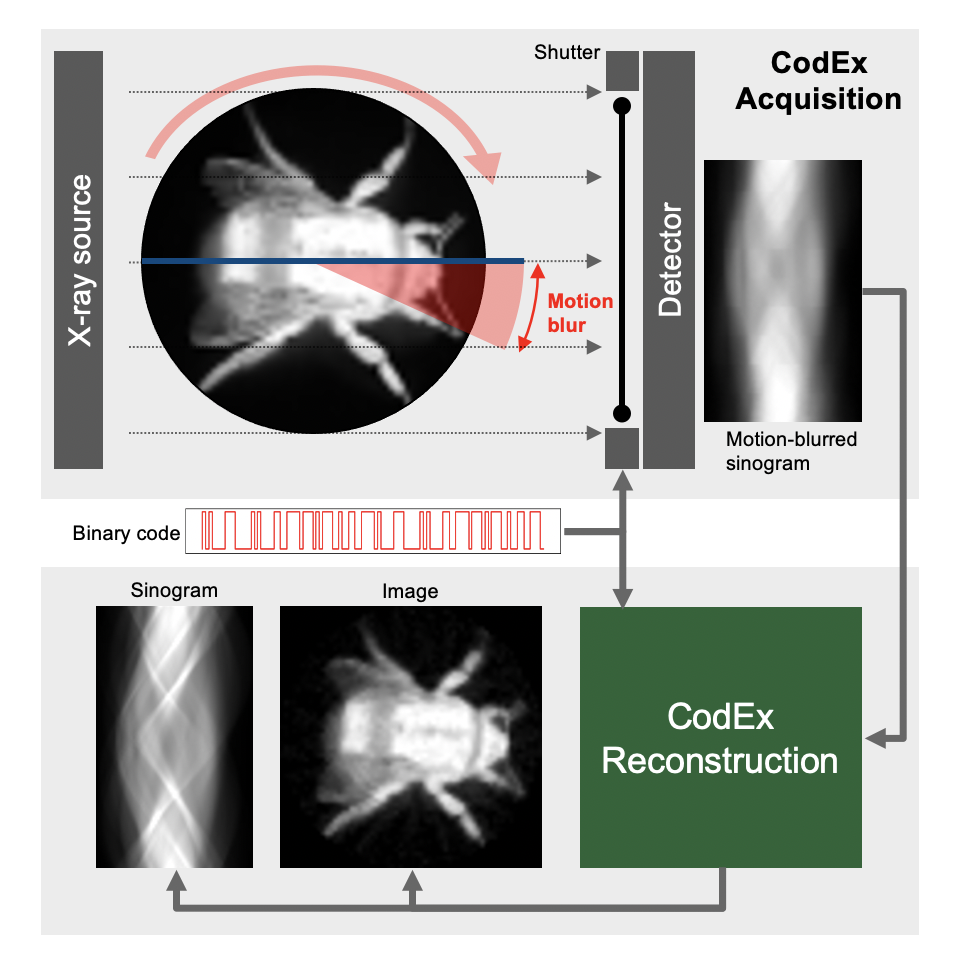}
\caption{Illustration of our method, CodEx. 
CodEx is a synergistic combination of a coded acquisition and a non-convex Bayesian reconstruction.
During acquisition, CodEx flutters the exposure on and off rapidly in a known binary code during each view, resulting in a coded motion blur.
CodEx reconstruction uses the knowledge of the code to effectively invert the coded motion blur.
A well chosen code can improve the accuracy of CodEx reconstruction, but CodEx can improve the reconstruction quality even without coding by inverting the motion blur due to fly-scanning.}
\label{fig:conops}
\end{figure}

In step-and-shoot acquisition views are taken at a set of discrete angles, with each view containing little or no blur.
This can be done by rotating and stopping the object before acquiring each view, or by rotating continuously and taking each view with a short exposure time that eliminates blur.
In either case, photons are wasted as the object is rotated.

A more practical acquisition strategy is fly-scanning~\cite{ching2018rotation} where the object rotates continuously and exposures are taken with 100\% duty cycle so that no photons are wasted.
In this case there are roughly two alternatives to acquiring the views, which we will refer to as slow-rotation and fast-rotation fly-scanning.

In slow-rotation fly-scanning, each view can be taken with a long exposure without introducing blur.
However, with long exposures, the total time required to take $M_\theta$ views is relatively long.
This is a problem if the object being imaged is moving or changing with time because the object motion will introduce artifacts in the reconstruction.
One might choose to reduce the value of $M_\theta$ in order to reduce the total acquisition time for all the views.
However, in this case, the acquired sparse views span only a limited total view angle $<< 180^\circ$ (i.e., limited view), resulting in significant reconstruction artifacts\cite{frikel2013characterization}.

In fast-rotation fly-scanning, which is the focus of this research, the object is rotated rapidly, and sparse view samples are collected using a relatively small value of $M_\theta$.
Since the object is rapidly rotating, the $M_\theta$ views can span the desired total view angle of $\approx 180^\circ$, thereby avoiding limited view angle reconstruction artifacts.
However, the disadvantage of fast rotation fly-scanning is that each view must be acquired over a significant blur angle. 
Consequently, conventional reconstruction methods produce blurry reconstructions when used with fast-rotation fly-scanning.

One obvious approach to reducing the effects of fly-scan motion blur is to perform linear deconvolution as a pre-processing step. However, an important limitation of this approach is that it requires dense view sampling. This is because sparse view sampling results in aliasing that can not be fully corrected using linear deconvolution.

Previous research based on sinogram pre-processing includes the work of Chang et al.~\cite{chang2011preliminary} who proposed pre-processing using a linear deblurring filter.
Later, Chen et al.~\cite{chen2015computed} proposed a more sophisticated pre-processing based on regularized sinogram deblurring using a TV regularizer.
However, both these methods assumed that multiple rotations of the object were used to collect a full set of dense views.
This is an important limitation since sparse view sampling allows for faster acquisition of objects that change with time.

Other researchers have proposed iterative reconstruction methods for mitigating motion blur caused by fast-rotation fly-scanning.
Cant et al.~\cite{cant2015modeling} performed SIRT~\cite{SIRT} reconstruction that incorporated a linear approximation to the blur.
However, this method required dense view sampling in order to avoid artifacts in the full field-of-view reconstruction.
More recently, Ching et al.~\cite{ching2019time} used ART~\cite{ART} reconstruction along with a linear approximation to the blur, and they also assumed binary coding during the acquisitions of each view in order to improve the invertiblity of the problem. Again, this approach assumed the availability of dense view samples in order to achieve full deconvolution.

Finally, Tilley et al.~\cite{steven2018high} performed model-based reconstruction using a non-linear forward model of the blur. While this approach allows for the possibility of sparse view reconstruction, it was only demonstrated for case were the views were densely sampled.

In this paper, we introduce CodEx, a modular framework for accurate non-linear tomographic reconstruction of sparse-view data acquired with fast-rotation fly-scans.
Figure~\ref{fig:conops} illustrates our approach.
During acquisition, CodEx modulates the X-ray flux by a known binary code during each view.
This results in a coded motion blur that is easily inverted~\cite{raskar2006coded}.

Modulation of the X-ray flux by a code can be performed at the X-ray source, shutter, or the detector.
For synchrotron or X-ray free electron laser (XFEL), the X-ray source can be pulsed~\cite{pergament2014high}.
Alternatively, the X-ray beam can be blocked using a rotary disc shutter~\cite{gembicky2005fast, rose2020variable}, or CMOS detectors can be switched on and off via an electronically generated gating signal~\cite{broennimann2006pilatus,ejdrup2009picosecond}.

CodEx subsequently uses the knowledge of the code to solve a non-convex iterative reconstruction problem using the alternating direction method of multipliers (ADMM) algorithm.
The resulting CodEx algorithm iteratively applies a non-linear deblurring operation followed by a model based iterative reconstruction (MBIR)~\cite{bouman1996unified} operation.
This results in a modular reconstruction algorithm that can be easily adapted to different CT geometries.

The novel contributions of this paper are:
\begin{enumerate}
    \item A modular reconstruction approach using ADMM for the nonlinear reconstruction of sparse-view transmission CT data;
    \item A coded exposure measurement scheme for collection of sparse transmission CT data;
    \item An interlaced view-sampling approach for progressively collecting sparse CT view data.
\end{enumerate}

We present results on simulated and binned experimental data in order to demonstrate the effectiveness of CodEx.
The results show CodEx produces the highest quality reconstruction results when the number of views and total photon dosage is limited.

\section{CT Fly Scanning}
\label{sec:problem_description}

In a conventional computed tomography (CT) setting, a single view is assumed to measure the projection of an object at a single angle.
To conform with this assumption, a step-and-shoot scanning is done where the object is rotated to each desired angle, a view measurement is taken, and the rotation is resumed.
However, this leads to slow acquisition and wasted photons.
A more practical approach is fly-scanning, where the object is continuously rotated while the view measurements are being taken.
However, in this case, the detector integrates the incoming photon-flux over a range of angles instead of a single angle, as shown in Figure~\ref{fig:blur}.
The resulting vector of $M_d$ expected photon-counts at the detector at angle $\theta_0$ can be written as
\begin{equation} \label{eq:blur_integration}
    \bar{\Lambda}^{\text{boxcar}}_{\theta_0} = \Lambda^0 \int_{\theta_0}^{\theta_0+\Delta\theta} \exp \left\{ -A_{\theta} \, x \right\} d\theta \ ,
\end{equation}
where $A_\theta \in \mathbb{R}^{M_d \times N}$ performs the forward projection of the image $x$ at angle $\theta$, $\Lambda^0$ is the photon-flux of the X-ray source per unit rotation angle, and $\Delta\theta$ is the blur-angle.
The super-script boxcar refers to the standard exposure where the detector shutter is on throughout the acquisition.

\begin{figure}[ht]
\centering     
\includegraphics[width=0.45\textwidth]{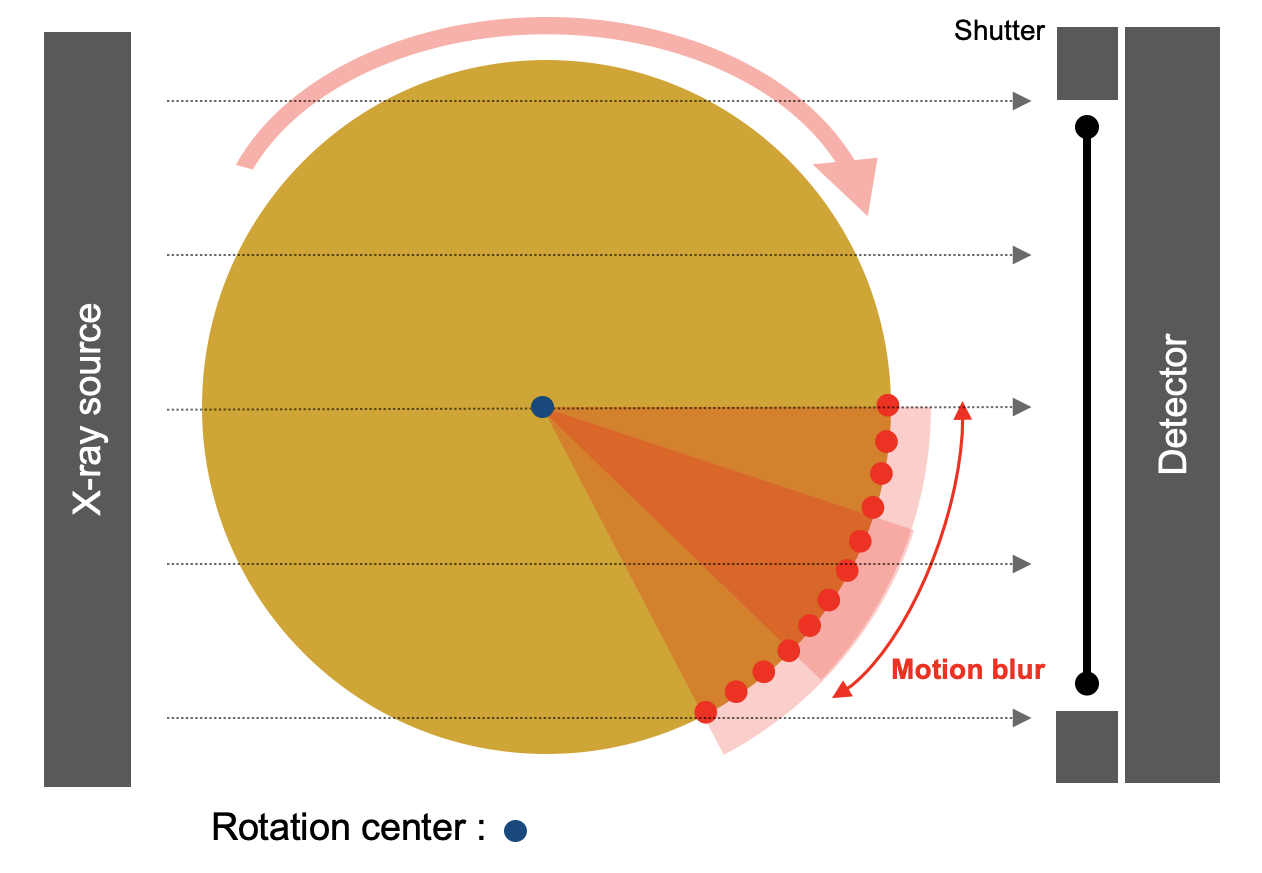}
\caption{Illustration of motion blur due to rotation.
In fly-scanning, the detector integrates photons over a range of projection angles (red sector). This integration can be approximated as a discrete sum over closely spaced micro-projection angles (red dots).
}
\label{fig:blur}
\end{figure}

We can approximate the integration in equation~\eqref{eq:blur_integration} using a discrete sum over $K$ closely spaced angles.
The resulting expected photon-counts, $\bar{\lambda}^{\text{boxcar}}_{\theta_0}$ can be written as
\begin{equation}\label{eq:blur_sum}
    \bar{\lambda}^{\text{boxcar}}_{\theta_0} = \lambda^0 \sum_{k=0}^{K-1} \exp \left\{ -A_{\theta_0+\frac{k\Delta\theta}{K} } \, x \right\} \ .
\end{equation}
We will refer to these $K$ closely spaced angles as micro-projection angles and the corresponding projections as micro-projections henceforth.

As $\Delta\theta$ becomes larger, we collect more photons in a single measurement, but this comes at the cost of blurred measurements, which leads to a reconstruction with motion artifacts.
To overcome this problem, our approach will be to modulate the photon flux at each of the $K$ micro-projection angles by a binary code.
In this case, the expected photon-counts at the detector at angle $\theta_0$ can be written as
\begin{equation}\label{eq:coded_sum}
    \bar{\lambda}_{\theta_0} = \lambda^0 \sum_{k=0}^{K-1} c_k \exp \left\{ -A_{\theta_0+\frac{k\Delta\theta}{K} } \, x \right\} \ ,
\end{equation}
where $c = [c_0, c_1, \cdots, c_{K-1}]$ is the binary code used to modulate the photons.
In this case, each measurement is formed by a coded sum over $K$ non-overlapping micro-projection angles.
Notice that equation~\eqref{eq:coded_sum} reduces to the boxcar case of equation~\eqref{eq:blur_sum} when $c = [1,1,\cdots, 1]$.
On the other hand, when $c = [1,0,\cdots, 0]$, i.e. one is pulsating the source or detector with a small exposure, equation~\eqref{eq:coded_sum} reduces to the step-and-shoot scanning case.
Modulating the photon flux results in loss of photons compared to the boxcar case of $c = [1,1,\cdots, 1]$, but many more photons are collected relative to the step-and-shoot case of $c = [1,0,\cdots, 0]$.
A good choice of code $c$ can result in an invertible blur while improving the signal to noise ratio (SNR) of the measurements relative to the step-and-shoot case.

\section{CodEx Formulation}
\label{sec:CodEx}

In this section, we will introduce CodEx, a synergistic combination of coded acquisition and CT reconstruction.
During the acquisition process, we collect $M_\theta$ 2D radiograph measurements at $M_\theta$ different measurement angles.
As Figure~\ref{fig:blur} illustrates, each radiograph measures the projection of the object across a range of angles which can be written as a non-linear function of $K$ micro-projections using the discrete approximation in equation~\eqref{eq:blur_sum}.
Some of the $M_\theta$ measurements can have overlapping projections and share the same micro-projections.
Let us define $N_\theta$ to be the number of unique micro-projection angles out of the $KM_\theta$ maximum possible micro-projection angles.

\subsection{Data Likelihood Model}

For each measurement, we use a binary code $c = [c_0, c_1, \cdots, c_{K-1}]$ to modulate the photon-flux over $K$ micro-projection angles as shown in equation~\eqref{eq:coded_sum}.
The resulting vector of expected photon counts for all the measurement angles can be written as
\begin{equation}
\label{eq:lambda_abstract}
    \bar{\lambda} = \bar{c} \lambda^0 C  \exp\left\{ -Ax \right\} \ ,
\end{equation}
where, $\bar{\lambda} \in \mathbb{R}^{M_\theta M_d}$ is the vector of expected photon counts for the $M_\theta$ measurement angles and $M_d$ detector pixels, $C \in \mathbb{R}^{M_\theta M_d \times N_\theta M_d}$ is a sparse matrix  that performs the coded sum in equation~\eqref{eq:coded_sum} and is normalized such that each row of $C$ sums to $1$,  $\bar{c} = \sum_{k=0}^{K-1}c_k$ is the normalizing constant, $A \in \mathbb{R}^{N_\theta M_d \times N}$ projects the image $x \in \mathbb{R}^{N}$ for the $N_\theta$ unique micro-projection angles.
The structure of the matrix $C$ depends on how the $M_\theta$ measurement angles are arranged, and how each measurement relates to the micro-projections at $K$ micro-projection angles.
Section~\ref{sec:interlaced_sampling} provides details on the structure of $C$ for a practical interlaced view-sampling strategy.

The vector of incident photon counts at the detector, $\lambda \in \mathbb{R}^{M_d M_\theta}$ are given by
\begin{equation}\label{eq:pois}
\lambda \sim \mbox{Pois} ( \bar{\lambda} ) \ ,
\end{equation}
where $\mbox{Pois} (\bar{\lambda})$ denotes an element wise Poisson distribution with mean $\bar{\lambda}$.
A higher expected photon-count $\bar{\lambda}$ results in a higher signal to noise ratio (SNR) in the measurements.
A good code $c$ introduces an invertible blur and allows us to preserve high frequency information while at the same time collecting sufficient number of photons~\cite{raskar2006coded}.

In order to derive the forward model, we first convert the photon-count measurements into projection measurements as is typically done in tomography.
In order to do this, we normalize by the ``blank scan'' obtained when the object is removed (i.e., $x=0$), and we take the negative log to form
\begin{align}
    y &= -\log\left\{ \dfrac{\lambda }{\bar{c} \lambda^0 } \right\} \ ,
\label{eq:YlogLL}
\end{align}
where $y \in \mathbb{R}^{M_d M_\theta}$ is the vector of projection measurements for all $M_\theta$ views, and $\bar{c}=\sum_{i=0}^{K-1} c_k$ results from the assumption that $x=0$ in equation~\eqref{eq:coded_sum}.
If the photon counts are large, then we can make the approximation~\cite{bouman1996unified,balke2018separable} that
\begin{eqnarray}
E\left[ y |x \right] &\approx& - \log\left\{ C \exp \left\{ -Ax \right\}  \right\} \notag \\
Var\left[ y |x \right] &\approx& D^{-1} \ .
\end{eqnarray}
where $D = \text{diag}\{ \lambda \}$.
In practice, the true photon counts, $\lambda$ are often unknown.
Consequently, we set the precision matrix $D$ as 
\begin{equation}
    D = \text{diag}\{ w \exp \left\{ - y \right\}  \} \ ,
\end{equation}
where the scalar $w$ is empirically chosen~\cite{mohan2015timbir,balke2018separable}.

Using a second order Taylor series approximation~\cite{bouman1996unified} to the Poisson log-likelihood function, we can write the log-likelihood function as
\begin{equation}
\label{eq:LL_abstract}
-\log p(y | x ) = \dfrac{1}{2}
\norm{ y + \log\left\{ C \exp \left\{ -Ax \right\}  \right\} }_{D}^2 + \text{const}(y) \ .
\end{equation}
Here $\text{const}(y)$ refers to constant terms that are not a function of $x$ and can thus be ignored while optimizing with respect to $x$.

\subsection{MAP Estimate}

The X-ray attenuation coefficient image $x^{*}$ can be reconstructed by computing a Maximum A Posteriori (MAP) estimate as
\begin{align}
x^{*} &= \argmin_{x} \left\{ -\log p(y|x) -\log p(x) \right\} 
\notag \\
&= \argmin_{x} \Bigg\{ \dfrac{1}{2}
\norm{ y + \log\left\{ C \exp \left\{ -Ax \right\}  \right\} }_{D}^2 + h(x) \Bigg\}
\label{eq:MAP_abstract}
\end{align}
where $h(x)$ is the regularization or prior model~\cite{majee20194d,majee2017model} and the data likelihood term $-\log p(y|x)$ follows from equation~\eqref{eq:LL_abstract}.
Notice that the non-linear logarithmic and exponential terms make direct optimization of the cost function in equation~\eqref{eq:MAP_abstract} challenging.
However, in the following sections, we propose a modular algorithm for solving (\ref{eq:MAP_abstract}) that makes the solution practical to implement.

\subsection{ADMM Formulation}

In order to simplify the optimization in equation~\eqref{eq:MAP_abstract}, we split the cost function into two parts with the following constraint
\begin{equation}
    p = Ax \ ,
\end{equation}
where $p \in \mathbb{R}^{N_\theta M_d} $ is the projection of the image $x$ at the $N_\theta$ finely spaced micro-projection angles.
In other words, $p$ is the full set of unobserved micro-projections of the object.
We thus form the following equivalent problem.
\begin{equation}
\begin{aligned}\label{eq:MAPSplit}
   x^*, p^* = & \argmin_{x,p} \Bigg\{ \dfrac{1}{2}
\norm{ y + \log\left\{ C \exp \left\{ -p \right\}  \right\} }_{D}^2 + h(x) \Bigg\} \\
    & \text{s.t} \ p = A x 
\end{aligned}
\end{equation}

Next, we will use the alternating directions method of multipliers (ADMM) method~\cite{boyd2011distributed} to solve the constrained optimization of equation~(\ref{eq:MAPSplit}).
The augmented Lagrangian for this problem is given by
\begin{align}\label{Eq:AugLag2}
    L(p,x,u) =& \dfrac{1}{2}\norm{ y + \log\left\{ C \exp \left\{ -p \right\}  \right\} }_{D}^2 + h(x)
    \notag \\
    &+ \dfrac{1}{2 \sigma^2} \norm{p - A \, x + u}^2  \text{,}
\end{align}
where $\sigma$ is a tunable parameter, and $u$ is the scaled dual variable.
The ADMM algorithm for this problem can then be formulated as Algorithm~\ref{algo:admm_updates2}.

\begin{algorithm}

\DontPrintSemicolon
Initialize: $p$, $x$, $u$

\While{not converged}
{
    $ p \leftarrow \argmin_{p} L(p,x,u) $ \;
    $ x \leftarrow \argmin_{x} L(p,x,u) $ \;
    $ u \leftarrow u + p - A x $ \;
}
$ x^{*} \leftarrow x $

\caption{ADMM formulation for coded exposure reconstruction}
\label{algo:admm_updates2}
\end{algorithm}

\subsection{Modular Implementation}

Note that the optimization sub-problems in Algorithm~\ref{algo:admm_updates2} can be simplified as
\begin{align}
    \argmin_{p} L(p,x,u)  = \argmin_{p} &\Big\{ \dfrac{1}{2}\norm{ y + \log\left\{ C \exp \left\{ -p \right\}  \right\} }_{D}^2 \notag \\
    &+ \dfrac{1}{2 \sigma^2} \norm{p - (Ax-u)}^2 \Big\} \text{,} 
    \label{eq:augLagrange_p}
    \\
    \argmin_{x} L(p,x,u)  = \argmin_{x} &\Big\{ \dfrac{1}{2 \sigma^2}\norm{(p+u)-Ax}^2 \notag \\
    &+ h(x) \Big\} \text{.}
\label{eq:augLagrange_x}
\end{align}

We can rewrite the optimization problems in equations~\eqref{eq:augLagrange_p}~and~\eqref{eq:augLagrange_x} in a more compact form as
\begin{align}
    \argmin_{p} L(p,x,u) &= F_d(Ax-u) \text{,}
\label{eq:augLagrange_p_compact}
    \\
    \argmin_{x} L(p,x,u) &= F_t(p+u) \text{,}
\label{eq:augLagrange_x_compact}
\end{align}
where the operators $F_d$ and $F_t$ are defined as
\begin{align}
    F_d(\tilde{p}) = \argmin_{p} &\Bigg\{ \dfrac{1}{2}\norm{ y + \log\left\{ C \exp \left\{ -p \right\}  \right\} }_{D}^2 \notag \\
    &+ \dfrac{1}{2 \sigma^2} \norm{p - \tilde{p}}^2 \Bigg\} \text{,}
\label{eq:op_F_d}
\\
    F_t(\tilde{p}) = \argmin_{x} &\Bigg\{ \dfrac{1}{2 \sigma^2}\norm{\tilde{p}-Ax}^2 +
    h(x) \Bigg\} \text{,}
\label{eq:op_F_t}
\end{align}
where $\tilde{p}$ is a representative variable.

Both operators $F_d$ and $F_t$ have intuitive interpretations.
From its form, function $F_d$ can be interpreted as the MAP deblurring function.
Intuitively, $F_d$ computes the MAP estimate of the micro-projections $p$ given the coded blurred measurements $y$ and a prior distribution of $N(\tilde{p},\sigma^2 I)$.
In other words, $F_d$ is a function that recovers the full set of unobserved micro-projections in the proximity of $\tilde{p}$.
On the other hand, the function $x=F_t(p)$ has the simple interpretation of being a function that computes the regularized tomographic reconstruction, $x$, given the micro-projections $p$.

\begin{algorithm}

\DontPrintSemicolon
Initialize: $p$, $x$, $u$

\While{not converged}
{
    $ p \leftarrow \tilde{F_d} (A x - u; p) $ \;
    $ x \leftarrow \tilde{F_t} (p + u; x) $ \;
    $ u \leftarrow u + p - A x $ \;
}
$ x^{*} \leftarrow x $

\caption{CodEx reconstruction algorithm}
\label{algo:coded_updates2}
\end{algorithm}

Algorithm~(\ref{algo:coded_updates2}) shows the complete CodEx reconstruction algorithm.
Since it is impractical to minimize the functions in equations \eqref{eq:op_F_d} and \eqref{eq:op_F_t} completely, we perform partial updates starting from an initial condition. 
Convergence of ADMM for partial update solutions to sub-problems have been demonstrated in the literature~\cite{aslan2019joint,sridhar2020distributed}.
The operator $\tilde{F_d}(\tilde{p}; p_{\text{init}})$ denotes the computation of $F_d(\tilde{p})$ for a fixed number of partial iterations starting from an initial condition of $p_{\text{init}}$.
Similarly, the operator $\tilde{F_t}(\tilde{p}; x_{\text{init}})$ denotes the computation of $F_t(\tilde{p})$ for a fixed number of partial iterations starting from an initial condition of $x_{\text{init}}$.

Algorithm~\ref{algo:F_d_gradientDescent} outlines the computation of the deblurring function, $\tilde{F_d} (\tilde{p}; p_{\text{init}})$ that performs a partial update minimization of equation~(\ref{eq:op_F_d}) starting from an initial value of $p_{\text{init}}$.
We use a gradient descent approach with a backtracking line-search~\cite{armijo1966minimization} to perform the optimization.
The cost function is denoted by
\begin{equation}
    f_d(p) = \dfrac{1}{2}\norm{ y + \log\left\{ C \exp \left\{ -p \right\}  \right\} }_{D}^2 + \dfrac{1}{2 \sigma^2} \norm{p - \tilde{p}}^2 \notag \text{.}
\end{equation}
The gradient of the cost function $f_d(p)$ is denoted by
\begin{align}
    g = & - \text{diag}\left( \exp \left\{ -p \right\} \right) C^{\top} \text{diag}\left( C \exp \left\{ -p \right\} \right)^{-1} D r \notag \\ 
    & + \dfrac{1}{\sigma^2} (p - \tilde{p}) \ ,
\label{eq:grad}
\end{align}
where the measurement residual $r$ is defined as
\begin{equation}
    r = y + \log\left\{ C \exp \left\{ -p \right\}  \right\}  \ .
\end{equation}
The cost function $f_d(.)$ is evaluated several times for each iteration in order to choose an appropriate step-size of $\eta$ starting from an initial step-size of $\eta_0$.
We halve $\eta$ until the Armijo–Goldstein condition~\cite{armijo1966minimization} of $f_d(p - \eta g) \leq f_d(p)-\eta \epsilon \norm{g}^2$ is met for a control parameter $\epsilon$.
Note that since $C$ performs the coded sum in equation~\eqref{eq:coded_sum} along the angle dimension, multiplication by both $C$ and $C^\top$ can be computed for each detector pixel independently.
Thus, the gradient computation in equation~\eqref{eq:grad} can be computed independently for each detector pixel, making the computation quite suitable for a parallel implementation.

\begin{algorithm}

\DontPrintSemicolon
\KwIn{Initial micro-projections: $ p_{\text{init}} $ \newline
Proximal micro-projections input: $\tilde{p}$}
\KwOut{Final micro-projections: $ p^{*} $}

Initialize: $p \gets p_{\text{init}} $ 

\For{$i \gets 1$ to $n_p$}   
{
    $r \gets y + \log\left\{ C \exp \left\{ -p \right\}  \right\} $ \\
    
    $\begin{aligned}
    g \gets & - \text{diag}\left( \exp \left\{ -p \right\} \right) C^{\top} \text{diag}\left( C \exp \left\{ -p \right\} \right)^{-1} D r\\
     & + \dfrac{1}{\sigma^2} (p - \tilde{p})
    \end{aligned}$ \\
        
    $\eta = \eta_0$ \;
    \While{$f_d(p - \eta g)>f_d(p)-\eta \epsilon \norm{g}^2$}
    {
        $\eta \gets \eta/2$ \;
    }
    $ p \gets p - \eta g $ \;
}
$ p^{*} \gets p $

\caption{Computation of the deblurring function $\tilde{F}_d$}
\label{algo:F_d_gradientDescent}
\end{algorithm}

Algorithm~\ref{algo:F_t} outlines the computation of the tomographic reconstruction function, $\tilde{F_t} (\tilde{p}; x_{\text{init}})$ that performs a partial update minimization of equation~(\ref{eq:op_F_t}) starting from an initial value of $x_{\text{init}}$.
The optimization can be performed using any off-the-shelf software module that can perform regularized inversion.
More implementation details are given in section~\ref{sec:results}.

\begin{algorithm}

\DontPrintSemicolon
\KwIn{Initial reconstruction: $ x_{\text{init}} $ \newline
Projections data: $\tilde{p}$ \newline
}
\KwOut{Final reconstruction: $ x^{*} $}

Initialize: $x \gets x_{\text{init}} $ 

\For{$i \gets 1$ to $n_t$}   
{
    $e \gets p-Ax $ \;
    $x \gets x + \text{Update}(e, A) $ \;
}
$ x^{*} \gets x $

\caption{Computation of the tomographic reconstruction function $\tilde{F}_t$}
\label{algo:F_t}
\end{algorithm}

Even though the coded deblurring and CT reconstruction problems are tightly coupled, the modular structure of Algorithm~(\ref{algo:coded_updates2}) separates them into deblurring and CT reconstruction sub-problems that must be solved repeatedly until convergence.
This allows us to use existing algorithms for iterative Bayesian tomographic reconstruction, which we believe is an important advantage of CodEx.

Note that the deblurring operator $F_d$ performs the optimization in equation~\eqref{eq:op_F_d} purely in the projection domain and is independent of the CT geometry.
Only the tomographic reconstruction operator $F_t$, and the forward projection operator $A$ in Algorithm~(\ref{algo:coded_updates2}) depend on the CT geometry under consideration. 
Our approach can therefore be easily extended to other CT geometries by incorporating a different reconstruction operator $F_t$, and a forward projection operator $A$ specific to that CT geometry.

\section{Interlaced View Sampling}
\label{sec:interlaced_sampling}

Recently proposed interlaced view sampling schemes~\cite{mohan2015timbir,zang2018space} have demonstrated improved reconstruction quality for time-resolved tomography compared to traditional progressive view sampling.
In interlaced view sampling, the view measurements are collected over multiple rotations of the object, rather than a single rotation as with progressive sampling.
This allows a wider range of angular measurements per unit time, thereby improving the reconstruction quality for time-resolved tomography~\cite{mohan2015timbir}.

Next we define the interlaced view sampling we will use for CodEx.
Let $M_\theta$ denote the number of acquired views with angles given by
\begin{equation}
    \theta_i = \frac{\pi i K }{N_\theta} \ \text{, for $i=0,\cdots,M_\theta-1$},
\label{eq:angles}
\end{equation}
where $K$ is the code length of the coded exposure, and $N_\theta$ is the number of unique micro-projection angles in $[0,\pi]$.
Using equation~\eqref{eq:angles}, we can ensure that every view angle is unique by choosing $K$ and $N$ to be relatively prime so that $\text{gcd}(K, N_\theta) = 1$.
Theorem~\ref{theorem_unique_angles} in Appendix~\ref{sec:apendix} proves this result.

When the total number of views is equal to the number of micro-projection angles, so that $M_\theta=N_\theta$, then we say that the view sampling is dense. 
In this case, all the $N_\theta$ densely-spaced micro-projection angles are sampled.

In order to ensure $\text{gcd}(K, N_\theta) = 1$, we choose $N_\theta$ in the following way
\begin{equation}
    N_\theta = mK - n \ ,
\label{eq:MN_coprime}
\end{equation}
where the positive integers $m$ and $n$ are chosen so that $K$ and $n$ are co-prime.
More specifically, Theorem~\ref{theorem_N_theta} in Appendix~\ref{sec:apendix} shows that when $\text{gcd}(K, n) = 1$ then $\text{gcd}(K, N_\theta) = 1$.

Once a suitable $n$ is chosen that is co-prime to $K$, $m$ can be adjusted to tune the angular spacing between view-angles.
For a small $n$, the view-angles $\theta_i$ in equation~\eqref{eq:angles} are roughly separated by $\frac{\pi}{m}$.

Some typical choices of the parameters $K$, $m$, $n$ is given in Table~\ref{table:angle_params}.

\begin{table}[h!]
\centering
\begin{tabular}{|l | c | c | c | r|} 
\hline
Code-length &&&& Blur-angle\\
 $K$ & $m$ & $n$ & $N_\theta = mK - n $ & $\Delta\theta = \frac{K 180^\circ}{N_\theta} $ \\ 
&&&&\\
\hline
$52$ & $2$ & $27$ & $77$ & $121.56^\circ$ \\ 
$52$ & $5$ & $27$ & $233$ & $40.17^\circ$ \\ 
$52$ & $10$ & $27$ & $493$ & $18.98^\circ$
\\
$52$ & $20$ & $27$ & $1013$ & $9.24^\circ$
\\
\hline
\end{tabular}
\\
\vspace{1mm}
\caption{Typical parameter choices for view-angle sampling}
\label{table:angle_params}
\end{table}

Fig~\ref{fig:view_angles} graphically illustrates the interlaced view-angle scheme in equation~\eqref{eq:angles} for the case when $m=5$, $n=5$, and $K=11$. 
In this case, there are $N_\theta = mK-n=50$ distinct micro-projection angles in the range of $[0,\pi]$.
Notice that each new measurement angle shown by a blue dot is spaced by $K=11$ discrete angular steps, but the measurement angles do not repeat until all $N_\theta = 50$ distinct micro-projection angles are sampled.

\begin{figure}[ht]
\centering     
\includegraphics[width=0.50\textwidth]{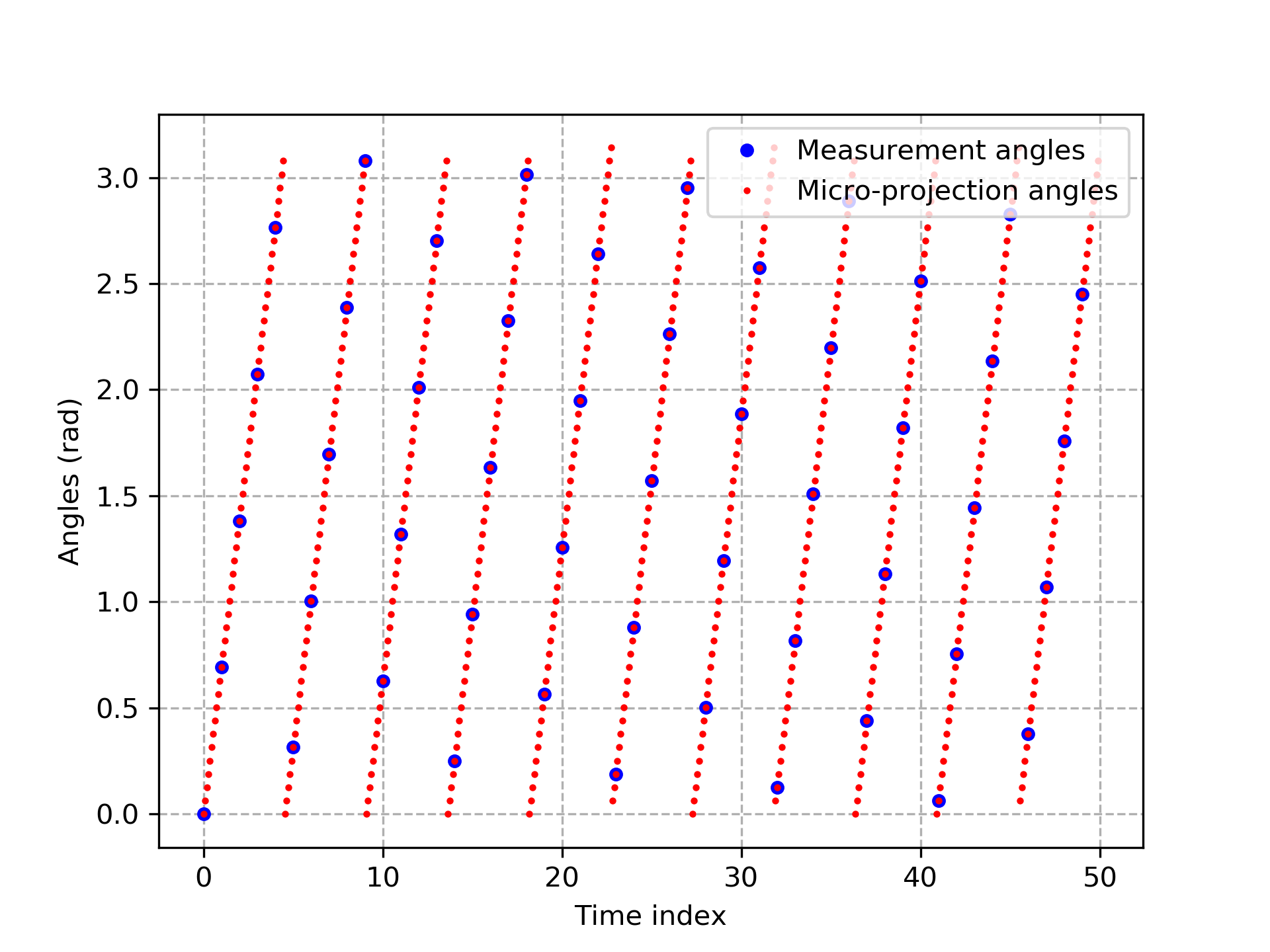}
\caption{Illustration of our View angle sampling for coded exposure CT for $N_\theta = mK-n = 50$ ($K=11$, $m=5$, $n=5$) and $M_\theta = 50$.
Each new measurement angle shown by a blue dot is spaced by $K=11$ discrete angular steps, but the measurement angles do not repeat until all $N_\theta = 50$ distinct micro-projection angles (blue dots) are sampled.
}
\label{fig:view_angles}
\end{figure}

Using the coded acquisition model in equation~\eqref{eq:coded_sum}, the expected photon-counts at the detector for the proposed interlaced views can be written as
\begin{align}
\label{eq:lambda_explicit}
    \bar{\lambda}_i = \lambda^0 \sum_{k=0}^{K-1} c_k \exp\left\{- \left(A_{\frac{\pi (i K +k) }{N_\theta} } \right) x\right\} & ,
    \notag
    \\
    \text{for $i=0,\cdots,M_\theta-1$} &,
\end{align}
where $\bar{\lambda}_i \in \mathbb{R}^{M_d}$ is the vector of expected photon-counts at the $M_d$ detector pixels, $\lambda^0$ is the photon-flux of the X-ray source, $c = [c_0, c_1, \cdots, c_{K-1}]$ is the binary code used to modulate the photons, and $A_{\frac{\pi i K }{N_\theta}} \in \mathbb{R}^{M_d \times N}$ performs the forward projection of the image $x$ at angle $\frac{\pi i K }{N_\theta}$.

Equation~\eqref{eq:lambda_explicit} reduces to the vectorized form in equation~\eqref{eq:lambda_abstract} with the coding matrix $C$ defined as
\begin{align}
    &C_{i(i_\theta, i_r, i_c ), j(j_\theta, j_r, j_c )} \\
    & = \begin{cases}
    \frac{c_k}{\bar{c}} & \text{if } \text{mod}(i_\theta K+k,N_\theta) = j_\theta \text{ for some } 0 \leq k < K \ , \\
    & \text{and } i_r=j_r \text{, and } i_c=j_c \\
    0 & \text{otherwise}
\end{cases} \ , \notag
\end{align}
where mod$()$ denotes the modulo operation and $i(i_\theta, i_r, i_c )$ represents the rasterized index $i$ as a function of the angular index $i_\theta$, row index $i_r$, and column index $i_c$.
Similarly, $j(j_\theta, j_r, j_c )$ represents the rasterized index $j$ as a function of the angular index $j_\theta$, row index $j_r$, and column index $j_c$.

\section{Results}
\label{sec:results}

We present results using simulated and binned physical data in order to demonstrate the effectiveness for our CodEx approach.
Our primary area of focus is parallel beam synchrotron imaging of objects found in material science and industrial CT~\cite{withers2021x}.
Accordingly, we consider parallel-beam sparse-view fast-rotation fly-sanning CT experiments representative of synchrotron imaging.

\subsection{Methods}

Below we describe the methods used in our experiments.
For the experimental results, we consider the three exposure codes:
\begin{itemize}
\item {\bf Snapshot code:} Uses code $(1,0,\cdots,0)$ of Figure~\ref{fig:code}(a);
\item {\bf Boxcar code:} This uses the box car code $(1,1,\cdots,1)$ of Figure~\ref{fig:code}(b).
\item {\bf Raskar code:} This uses the length 52 Raskar code~\cite{raskar2006coded} of Figure~\ref{fig:code}(c).
\end{itemize}
Notice that the snapshot code is equivalent to taking a very short exposure snapshot in that it freezes motion but limits the photon count.
When a longer code was required, we simply concatenated the Raskar code.
However, in practice code might be custom designed to optimize performance for CodEx reconstruction.

\begin{figure*}[!ht]
\centering     
\subfigure[Snapshot code]{\includegraphics[trim={2cm
0 2cm 0},clip,width=0.3\textwidth]
{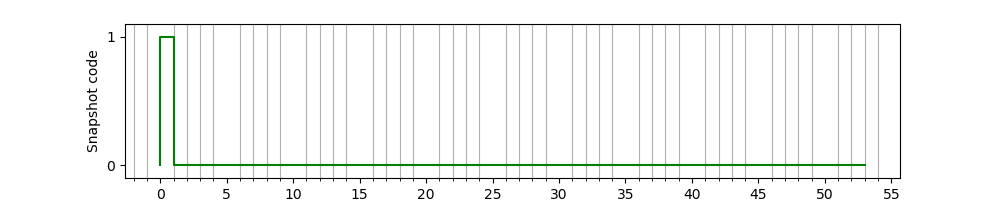}}
\subfigure[Boxcar code]{\includegraphics[trim={2cm
0 2cm 0},clip,width=0.3\textwidth]
{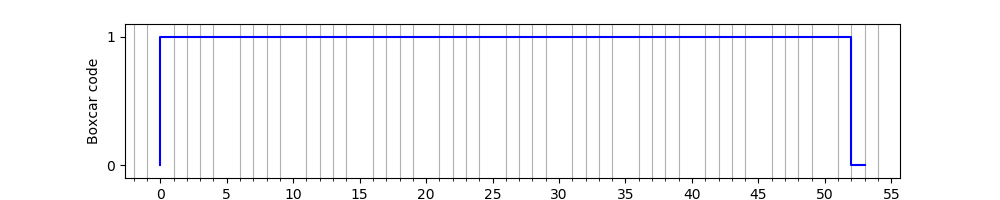}}
\subfigure[Raskar code]{\includegraphics[trim={2cm
0 2cm 0},clip,width=0.3\textwidth]
{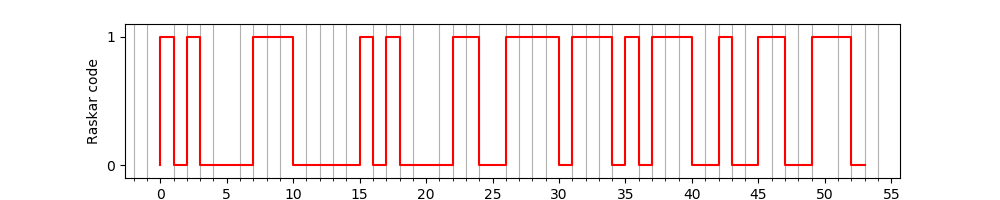}}
\caption{Examples of the binary codes of length $52$ that are used for modulating the photon-flux in our experiments.
}
\label{fig:code}
\end{figure*}

As baselines for comparison, we consider three alternative reconstruction methods.
\begin{itemize}
    \item {\bf MBIR:} Conventional MBIR reconstruction without deblurring of the measured views using photon weighting.
    \item {\bf IFBP:} Linear interpolation/deblurring of the sparse view data followed by FBP reconstruction.
    \item {\bf CodEx:} Full nonlinear CodEx reconstruction as described in the paper.
\end{itemize}
Note that the optimal linear interpolation of the views is done by solving the least squares problem
\begin{equation}
    p^*_{\text{deblur}} =  \argmin_{p} \Bigg\{ \dfrac{1}{2} \norm{ y - C p}^2  \Bigg\} \ .
\label{eq:serial_deblur}
\end{equation}

All MBIR reconstructions, including those used in the CodEx algorithm, are performed performed using the \mbox{SVMBIR} open-source code package \cite{svmbir-2020}.
The svmbir implementation uses a Markov random field based regularization and the reconstruction is computed using a cache optimized iterative coordinate descent~\cite{wang2016fast,wang2017massively}.

All tomographic data simulations are performed assuming the nonlinear forward model, the specified code, and Poisson photon counting statistics when noise is included.
This is done regardless of the reconstruction method.
When computing the CodEx reconstruction, we run $n_p = 5$ partial updates for computing the $\tilde{F}_d$ step and $n_t = 5$ partial updates for computing the $\tilde{F}_t$ step in Algorithm~\ref{algo:coded_updates2}.
For computing the baseline MBIR reconstruction, we use $400$ iterations.

Measurements of the modulation transfer function (MTF) metric~\cite{friedman2013simple} along the tangential and radial directions are performed using the ``Siemens star'' and ``concentric circle'' phantoms of Figures~\ref{fig:simresults_MTF_ang}(a) and~\ref{fig:simresults_MTF_rad}(a), respectively.
In both cases, we compute the MTF via the following steps: Take a line profile of the reconstruction perpendicular to an edge to obtain the edge spread function (ESF); Differentiate the ESF to obtain the line spread function (LSF); Compute the fast fourier transform (FFT) of the Hamming windowed LSF; Obtain the MTF by taking the absolute value of the FFT output.

\begin{figure*}[!ht]
\centering     
\subfigure[Phantom]{\includegraphics[clip,width=0.135\textwidth]
{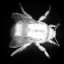}}
\subfigure[MBIR-snapshot ]{\includegraphics[clip,width=0.135\textwidth]
{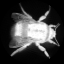}}
\subfigure[MBIR-boxcar ]{\includegraphics[clip,width=0.135\textwidth]
{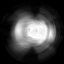}}
\subfigure[MBIR-Raskar ]{\includegraphics[clip,width=0.135\textwidth]
{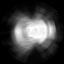}}
\subfigure[CodEx-snapshot ]{\includegraphics[clip,width=0.135\textwidth]
{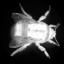}}
\subfigure[CodEx-boxcar ]{\includegraphics[clip,width=0.135\textwidth]
{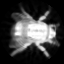}}
\subfigure[CodEx-Raskar ]{\includegraphics[clip,width=0.135\textwidth]
 {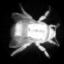}}
\caption{Comparison of reconstruction quality for simulated data without noise.
Experimental parameters are in Table~\ref{table:bumblebee_noiseless_setup}.
Each reconstruction uses either MBIR or CodEx reconstruction and uses one of three possible codes: snapshot, boxcar, or Raskar.
Notice that MBIR produces a blurred image when used with boxcar or Raskar codes.
In contrast, CodEx produces a sharp reconstruction in these two cases. 
MBIR and CodEx are equivalent when the snapshot code is used.
Also notice that the CodEx-Raskar reconstruction is sharper than the CodEx-boxcar reconstruction.
}
\label{fig:simresults_bumblebee_noiseless}
\end{figure*}

\subsection{Simulated Data Without Noise}

In this section, we perform reconstructions from simulated experiments without Poisson noise.
Details of the experimental parameters are given in Table~\ref{table:bumblebee_noiseless_setup}.

\begin{table}[!ht]
\centering{} 
\small
\begin{tabular}{r|l}
\toprule
Dosage: $\lambda^0$ & $\infty$ \\
Unique Micro-projection Angles: $N_{\theta}$  & 233 \\
Number of Views: $M_{\theta}$  & 233 \\
Code Length: $K$ & 52 \\
Blur Angle & 40.17$^\circ$ \\
Angular View Span & 25.89 rotations \\
Reconstruction Shape & 64$\times$64 \\
\bottomrule
\hline
\end{tabular}
\\
\vspace{1mm}
\caption{Parameters for simulated experiment without noise}
\label{table:bumblebee_noiseless_setup}
\end{table}

Figure~\ref{fig:simresults_bumblebee_noiseless} shows a qualitative comparison of CodEx with the baseline MBIR with different coded exposures.
Notice that MBIR produces a very blurry image when used with boxcar or Raskar codes.
This is because each view is acquired over a $40.17^\circ$ angle.
In contrast, CodEx produces a sharp reconstruction in these two cases. 
MBIR and CodEx are equivalent when the snapshot code is used since this code results in a very short exposure the freezes the rotation motion.
Importantly, the the CodEx-Raskar reconstruction is sharper than the CodEx-boxcar reconstruction.
The CodEx-boxcar reconstruction has noticeable angular blur, especially further from the center of rotation, which is reduced in the CodEx-Raskar reconstruction.

Figure~\ref{fig:simresults_MTF_ang} and Figure~\ref{fig:simresults_MTF_rad} compare the MTF along the tangential and radial directions for the boxcar and Raskar codes using CodEx reconstruction.
Each MTF curve is computed for both a small and large radial value.
Notice that the tangential MTF of CodEx-Raskar is much better than CodEx-boxcar.
This is reasonable since the effect of object rotation of blur increases with distance from the center of rotation.
Alternatively, the MTF in the radial direction is similar for both codes.

\begin{figure*}[!ht]
\raggedleft
\subfigure[Phantom]{\includegraphics[clip,width=0.19\textwidth]
{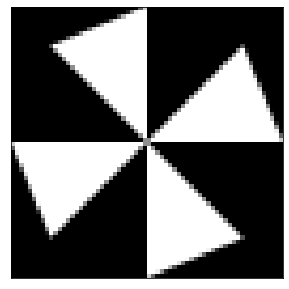}}
\subfigure[CodEx-boxcar reconstruction
]{\includegraphics[clip,width=0.19\textwidth]
{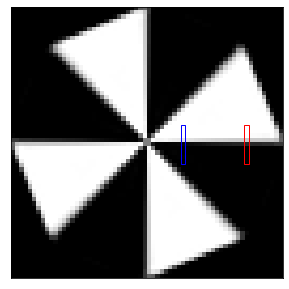}}
\subfigure[CodEx-boxcar MTF
]{\includegraphics[clip,width=0.19\textwidth]
{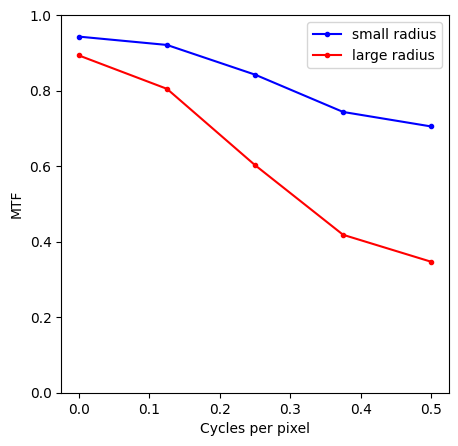}}
\subfigure[CodEx-Raskar reconstruction]{\includegraphics[clip,width=0.19\textwidth]
{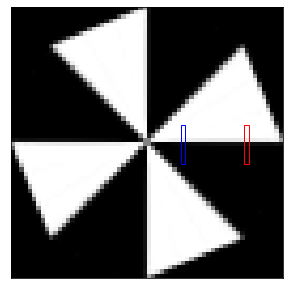}}
\subfigure[CodEx-Raskar MTF]{\includegraphics[clip,width=0.19\textwidth]
{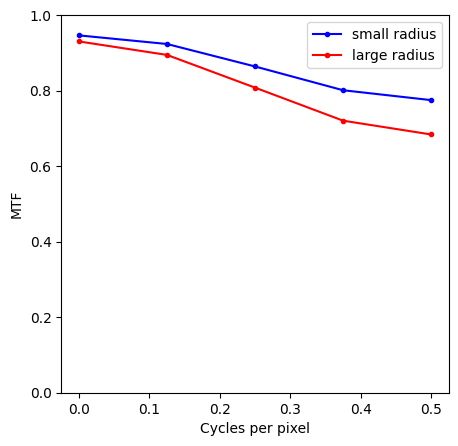}}
\caption{The tangential MTF for CodEx-Raskar and CodEx-boxcar reconstructions.
In comparison to CodEx-Raskar, the MTF value for CodEx-boxcar has a sharper decrease at higher frequencies.
Furthermore, the change in the MTF curve is higher for CodEx-boxcar compared to CodEx-Raskar.
These indicate better resolution of CodEx-Raskar in the tangential direction compared to CodEx-boxcar.
}
\label{fig:simresults_MTF_ang}
\end{figure*}

\begin{figure*}[!ht]
\raggedleft
\subfigure[Phantom]{\includegraphics[clip,width=0.19\textwidth]
{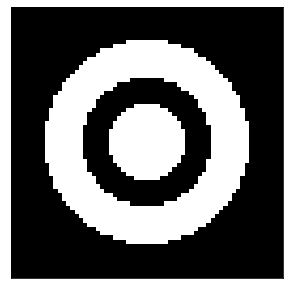}}
\subfigure[CodEx-boxcar reconstruction
]{\includegraphics[clip,width=0.19\textwidth]
{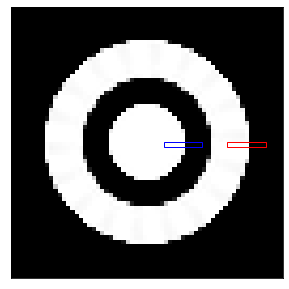}}
\subfigure[CodEx-boxcar MTF
]{\includegraphics[clip,width=0.19\textwidth]
{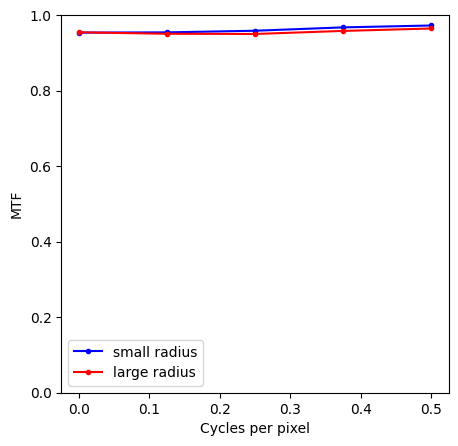}}
\subfigure[CodEx-Raskar reconstruction]{\includegraphics[clip,width=0.19\textwidth]
{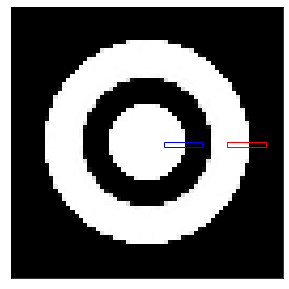}}
\subfigure[CodEx-Raskar MTF]{\includegraphics[clip,width=0.19\textwidth]
{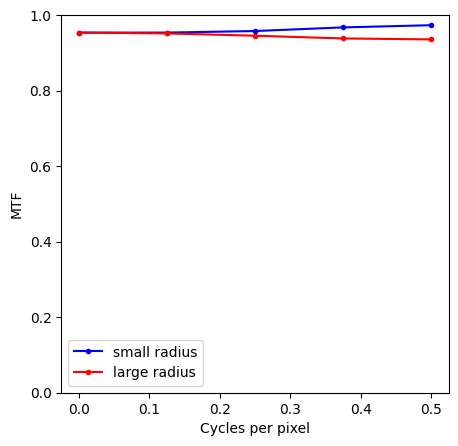}}
\caption{The radial MTF for CodEx-Raskar and CodEx-boxcar reconstructions.
Both CodEx-Raskar and CodEx-boxcar reconstructions exhibit similar behaviour in the MTF curves indicating similar resolution along the radial direction.
}
\label{fig:simresults_MTF_rad}
\end{figure*}

\subsection{Simulated Data With Poisson noise}

\begin{figure*}[!ht]
\subfigure[Phantom]{\includegraphics[clip,width=0.135\textwidth]
{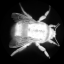}}
\subfigure[MBIR-snapshot
]{\includegraphics[clip,width=0.135\textwidth]
{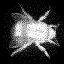}}
\subfigure[MBIR-boxcar ]{\includegraphics[clip,width=0.135\textwidth]
{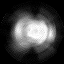}}
\subfigure[MBIR-Raskar
]{\includegraphics[clip,width=0.135\textwidth]
{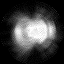}}
\subfigure[CodEx-snapshot ]{\includegraphics[clip,width=0.135\textwidth]
{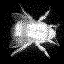}}
\subfigure[CodEx-boxcar
]{\includegraphics[clip,width=0.135\textwidth]
{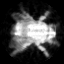}}
\subfigure[CodEx-Raskar ]{\includegraphics[clip,width=0.135\textwidth]
{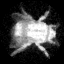}}
\caption{Comparison of reconstruction quality for simulated data with photon noise. 
Experimental parameters are given in Table~\ref{table:noisy_setup}.
For Raskar and boxcar codes, the baseline MBIR reconstruction leads to a blurred image.
In contrast, CodEx-Raskar and CodEx-boxcar reconstructions do not suffer from severe blurring artifacts.
The reconstructions with snapshot code suffer from high noise due to the limited photon counts in the measurements.
The reconstructions with boxcar code suffer from radial blur artifacts and loss of fine features due to the non-invertible nature of the blur kernel.
}
\label{fig:simresults_noisy}
\end{figure*}

\begin{figure}[!ht]
\centering
\subfigure[Primal Residual: RMSE($Ax^t$, $p^t$)]{\includegraphics[clip,width=0.22\textwidth]
{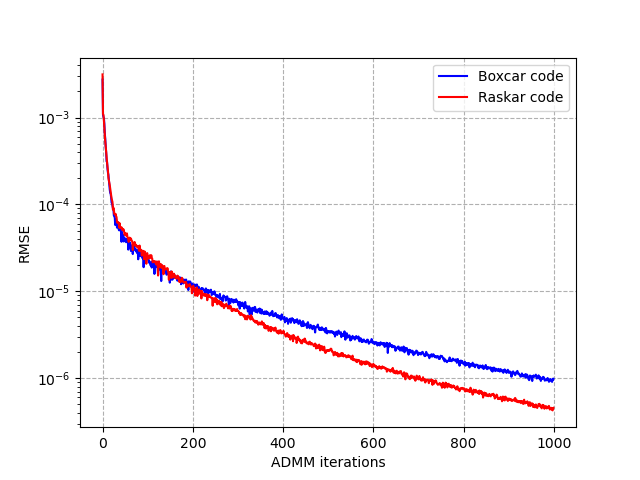}}
\subfigure[Dual Residual: RMSE($Ax^t$, $Ax^{t-1}$)]{\includegraphics[clip,width=0.22\textwidth]
{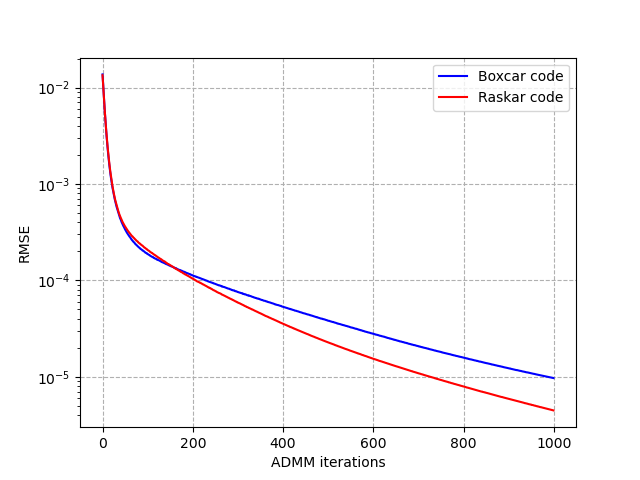}}
\caption{Primal and dual residual convergence plots for simulated data with photon noise.
Here $t$ refers to the ADMM iteration number.
The Raskar code leads to slightly improved convergence than the boxcar code.
}
\label{fig:convplots_noisy}
\end{figure}

\begin{figure*}[!ht]
\centering
\subfigure[Code length $52$ (Blur angle $40.17^{\circ}$)]{\includegraphics[clip,width=\textwidth]
{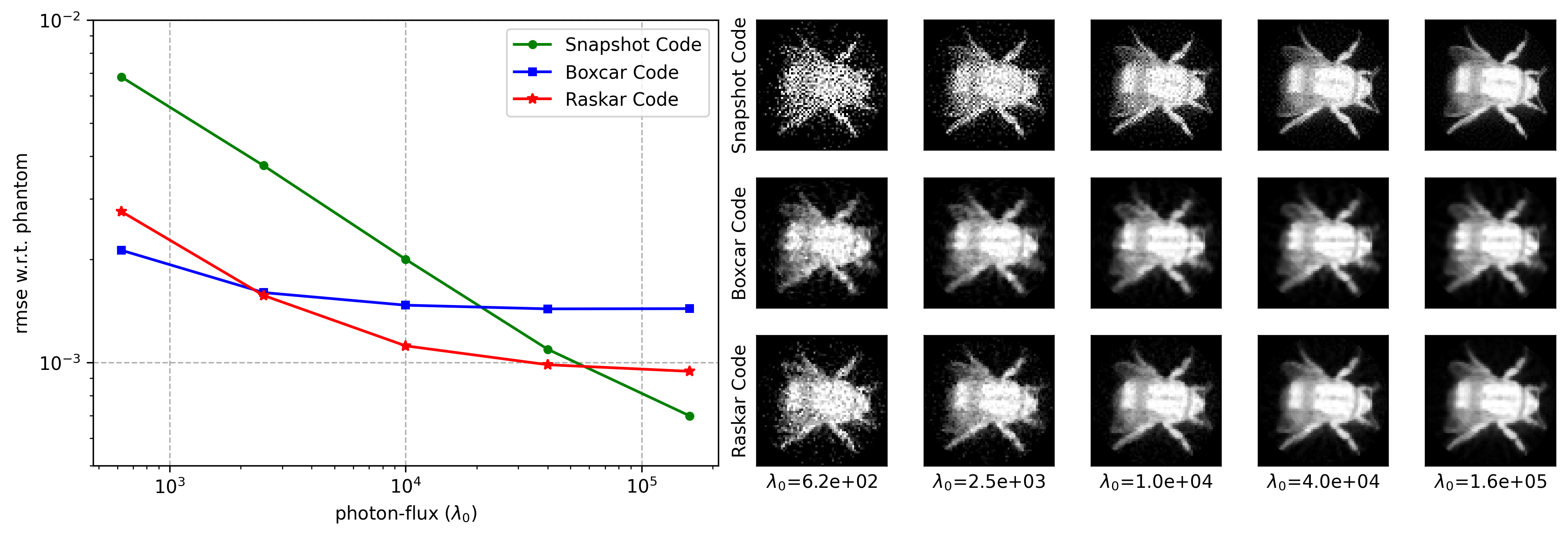}}
\subfigure[Code length $104$ (Blur angle $80.34^{\circ}$)]{\includegraphics[clip,width=\textwidth]
{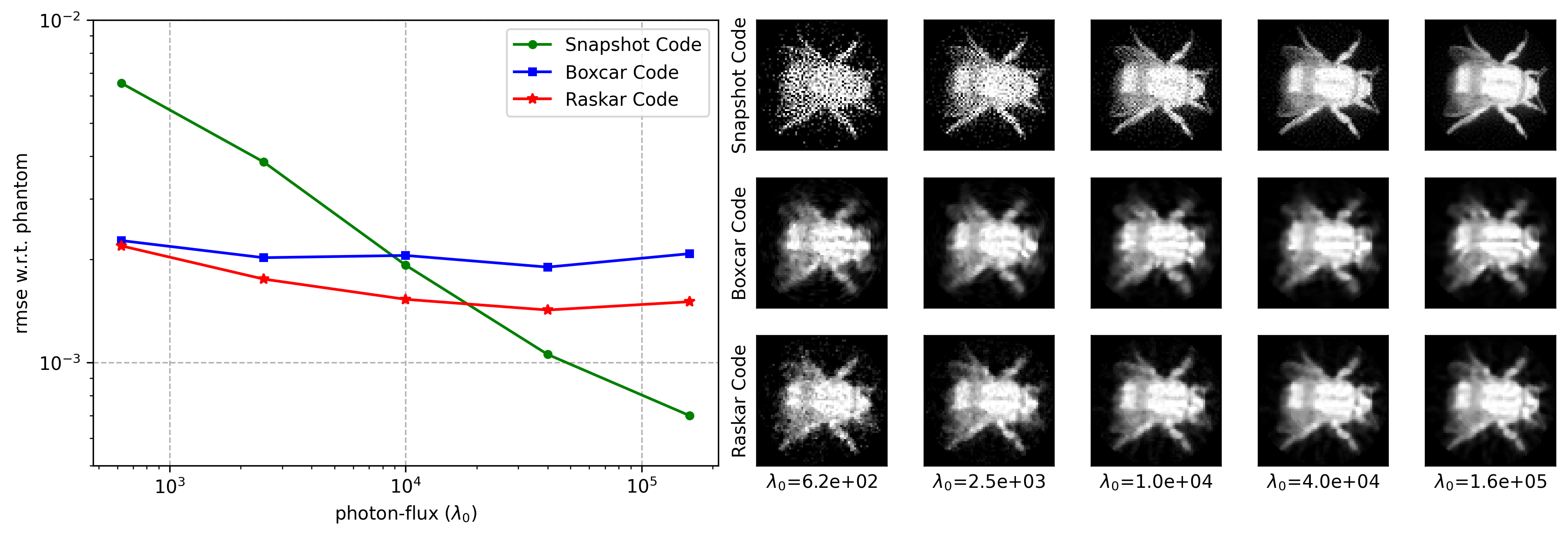}}
\subfigure[Code length $208$ (Blur angle $160.68^{\circ}$)]{\includegraphics[clip,width=\textwidth]
{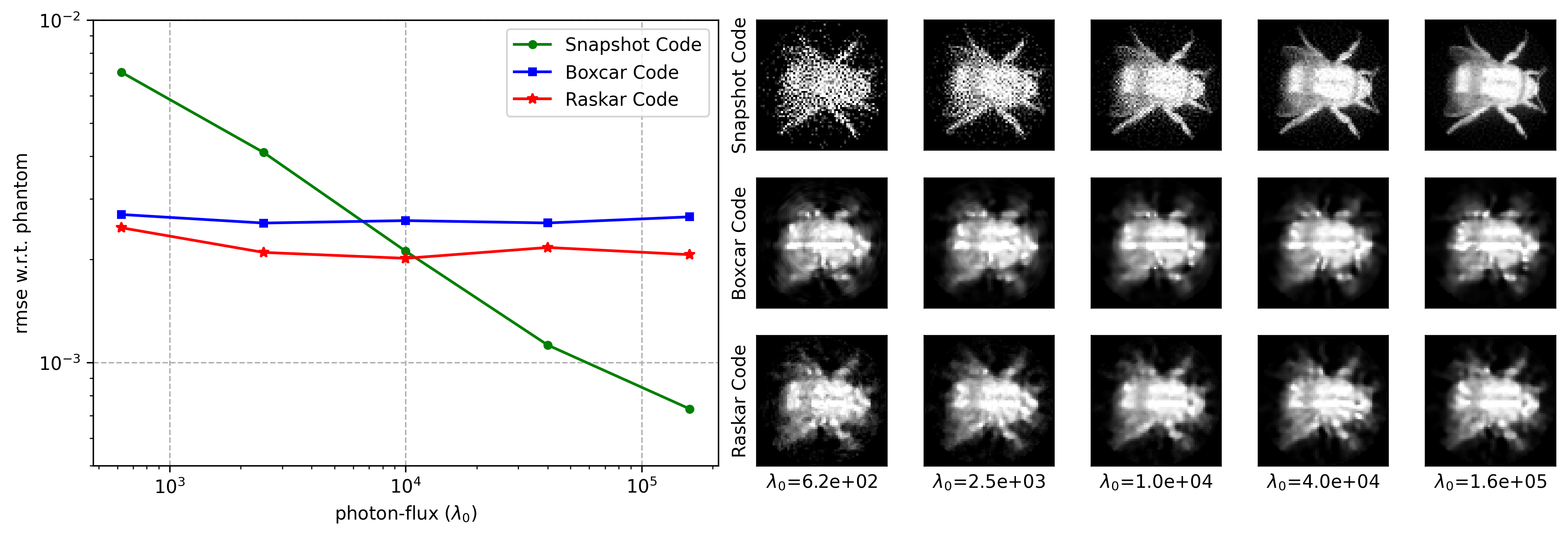}}
\caption{Effect of the photon-flux ($\lambda^0$) and the code length on reconstruction quality for different codes.
The remaining experimental parameters are kept the same as Table~\ref{table:noisy_setup}.
When the photon-flux ($\lambda^0$) is low, the boxcar code and Raskar code produce better image quality than the snapshot code as a result of collecting more photons.
When the photon-flux ($\lambda^0$) is high, the minor gains from increased photon count do not fully compensate for the loss of information by inverting the coded blur.
Consequently at high photon-flux, the snapshot code produce better image quality than the boxcar code and Raskar codes.
}
\label{fig:noise_plot}
\end{figure*}

In this section, we perform simulated experiments with more realistic Poisson noise using equations~\eqref{eq:lambda_abstract} and~\eqref{eq:pois}.
The experimental parameters are summarized in Table~\ref{table:noisy_setup}.

\begin{table}[!ht]
\centering{} 
\small
\begin{tabular}{r|l}
\toprule
Dosage: $\lambda^0$ & 10,000 \\
Unique Micro-projection Angles: $N_{\theta}$  & 233 \\
Number of Views: $M_{\theta}$  & 100 \\
Code Length: $K$ & 52 \\
Blur Angle & 40.17$^\circ$ \\
Angular View Span & 11.05 rotations \\
Reconstruction Shape & 64$\times$64 \\
\bottomrule
\hline
\end{tabular}
\\
\vspace{1mm}
\caption{Parameters for simulated experiment with noise}
\label{table:noisy_setup}
\end{table}

Figure~\ref{fig:simresults_noisy} shows a qualitative comparison of CodEx with the baseline MBIR for different coded exposures and Poisson noise.
Notice that these results are similar to those without noise, with the MBIR-boxcar and MBIR-Rashar reconstruction being very blurry.
However, in this case, the snapshot code yeilds a very noisy reconstruction for both the MBIR and CodEx reconstructions. This illustrates the weakness of the snapshot acquistion, i.e., that it wastes photons by using a very short exposure time.
So in this case, the CodEx-Raskar result has the advantage of lower noise while maintaining detail.

In Figure~\ref{fig:convplots_noisy}, we plot the primal residual, RMSE($Ax^t$, $p^t$) and dual residual, RMSE($Ax^t$, $Ax^{t-1}$) \cite{boyd2011distributed} at each ADMM iteration to illustrate the convergence of our method. 
Here $t$ refers to the ADMM iteration number.
Here we see that the Raskar code leads to a slightly improved convergence speed relative to the boxcar code.

In Figure~\ref{fig:noise_plot}, we show the effect of the photon-flux ($\lambda^0$) and the code length on reconstruction quality for different codes.
For each code type (snapshot, boxcar, and Raskar), we vary the photon-flux ($\lambda^0$) and plot the reconstruction RMSE relative to ground truth as a measure image quality.
We also show the resulting image for a visual inspection of image quality.
We repeat this process for each code lengths of $52$, $104$, and $208$.\footnote{The Raskar codes were extended by concatenating the basic length 52 code.}
The remaining experimental parameters are listed in Table~\ref{table:noisy_setup}.

Notice that when the photon-flux ($\lambda^0$) is low, the CodEx-boxcar and CodEx-Raskar produce better image quality than the CodEx-snapshot because they collect more photons.
Moreover, when the low photon-flux is low, CodEx-Raskar method benefits from the longer blur angle resulting from a longer code because it can collect more photons.
Alternatively, when the photon-flux ($\lambda^0$) is high, the minor gains from increased photon count do not fully compensate for the losses incurred when inverting the coded blur.
In this case, the MBIR or CodEx result using the snapshot code is superior to the CodEx-boxcar or CodEx-Raskar reconstructions.

\subsection{Short Duration Scan with Poisson Noise}

In this section we simulated a short duration scan in order to quantify the advantages of CodEx when scanning a dynamic object. To do this, we consider three scanning acquisition methodologies:
\begin{itemize}
    \item Slow-scan with parameters shown in~Table~\ref{table:dynamic_slow};
    \item Fast-scan with parameters shown in~Table~\ref{table:dynamic_fast};
    \item Coded-scan with parameters shown in~Table~\ref{table:dynamic_coded}.
\end{itemize}
Each scan mode is assumed to acquire either 20 or 40 views. We note that sparse view sampling is often used because it reduces motion artifacts by reducing the total scan time. The scan parameters are also adjusted so that the total photon count per view is constant. This is consistent with a physical experiment that uses a fixed scan time per view and a constant flux X-ray source.

Intuitively, the slow scan assumes that the object is rotated slowly, so the blur for each view is very small. However, the slow rotation means that the angular span is much less than the required $180^\circ$. Consequently, the slow scan does not cause blur, but it suffers from severe limited view artifacts.

The fast scan and coded scan both assume that the object is rotated rapidly. In this case, the views are blurred by $9.24^\circ$, but the angular span of the views greater than or equal to $180^\circ$. Consequently, the fast scan creates blurred view data, but it does not suffer from limited view artifacts.

\begin{table}[!ht]
\centering{} 
\small
\begin{tabular}{r|l}
\toprule
$\lambda^0$ & 520,000 \\
Unique Micro-projection Angles: $N_{\theta}$  & 1,013 \\
Number of Views: $M_{\theta}$  & 20 or 40 \\
Code Length: $K$ & 1 \\
Code Type & Snapshot \\
Blur Angle & 0.18$^\circ$ \\
Angular View Span & 3.38$^\circ$ or 6.93$^\circ$ \\
Reconstruction Shape & 128$\times$128 \\
\bottomrule
\hline
\end{tabular}
\\
\vspace{1mm}
\caption{Slow Scan Parameters}
\label{table:dynamic_slow}
\end{table}

\begin{table}[!ht]
\centering{} 
\small
\begin{tabular}{r|l}
\toprule
$\lambda^0$ & 10,000 \\
Unique Micro-projection Angles: $N_{\theta}$  & 1,013 \\
Number of Views: $M_{\theta}$  & 20 or 40 \\
Code Length: $K$ & 52 \\
Code Type & Boxcar \\
Blur Angle & 9.24$^\circ$ \\
Angular View Span & 175.56$^\circ$ or 360.36$^\circ$ \\
Reconstruction Shape & 128$\times$128 \\
\bottomrule
\hline
\end{tabular}
\\
\vspace{1mm}
\caption{Fast Scan Parameters}
\label{table:dynamic_fast}
\end{table}

\begin{table}[!ht]
\centering{} 
\small
\begin{tabular}{r|l}
\toprule
$\lambda^0$ & 10,000 \\
Unique Micro-projection Angles: $N_{\theta}$  & 1,013 \\
Number of Views: $M_{\theta}$  & 20 or 40 \\
Code Length: $K$ & 52 \\
Code Type & Raskar \\
Blur Angle & 9.24$^\circ$ \\
Angular View Span & 175.56$^\circ$ or 360.36$^\circ$ \\
Reconstruction Shape & 128$\times$128 \\
\bottomrule
\hline
\end{tabular}
\\
\vspace{1mm}
\caption{Coded Scan Parameters}
\label{table:dynamic_coded}
\end{table}

Figure~\ref{fig:simresults_dynamic} compares the results of 20 and 40 view scans using the following combinations of scan modes and reconstruction algorithms:
\begin{itemize}
    \item slow-MBIR: slow scan acquisition using MBIR reconstruction;
    \item fast-MBIR: fast scan acquisition using MBIR reconstruction;
    \item fast-IFBP: fast scan acquisition using IFBP reconstruction;
    \item fast-CodEx: fast scan acquisition using CodEx reconstruction;
    \item coded-CodEx: coded scan acquisition using CodEx reconstruction;
\end{itemize}
For each scanning/reconstruction combination we report the normalized root mean squared error (NRMSE) with respect to the phantom where the NRMSE of an image $x$ with respect to the phantom $x^0$ is defined as $\norm{x-x^0}/\norm{x^0}$.
Figure~\ref{fig:simresults_phantom_dynamic} shows the phantom used for the simulated experiments.

Notice that the slow scans of Figure~\ref{fig:simresults_dynamic}(a) and~(f) both have very poor quality even though they use MBIR reconstruction. This is because they have insufficient angular view range. Also notice that the fast scans in Figure~\ref{fig:simresults_dynamic}(b) and~(g) have poor quality. This is because the naive application of MBIR does not account for the blur caused by the object's fast rotation.

Figure~\ref{fig:simresults_dynamic}(c) and~(h) show the result of least squares linear deblurring followed by filtered back projection. This result is better than naive reconstruction with MBIR. However, this approach has two limitations. First, since the views are sparsely sampled, the blurring can not be fully removed due to aliasing. Second, the reconstruction does not account for the nonlinearity of the CT forward model.

Figure~\ref{fig:simresults_dynamic}(d), (i), (e) and~(j) compare the result of CodEx with the boxcar and Raskar codes using 20 and 40 views. 
Notice that with 20 views, the boxcar code is slightly better and with 40 views the Raskar code is slightly better. 
Intuitively, this is because with 20 views, the Raskar code wastes some photons. 
However, with 40 views, the angular view range increases to $360.36^\circ$, and the redundant view information along with the Raskar code can be exploited to recover additional resolution.

\begin{figure*}[!ht]
\centering     
\raggedleft
\subfigure[20 view slow-MBIR (NRMSE=0.5967) ]{\includegraphics[clip,width=0.19\textwidth]
{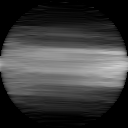}}
\subfigure[20 view fast-MBIR (NRMSE=0.1765) ]{\includegraphics[clip,width=0.19\textwidth]
{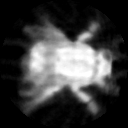}}
\subfigure[20 view fast-IFBP (NRMSE=0.1774) ]{\includegraphics[clip,width=0.19\textwidth]
 {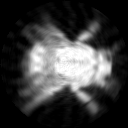}}
\subfigure[20 view fast-CodEx (NRMSE=0.1556)]{\includegraphics[clip,width=0.19\textwidth]
{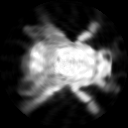}}
\subfigure[20 view coded-CodEx (NRMSE=0.1605)
]{\includegraphics[clip,width=0.19\textwidth]
{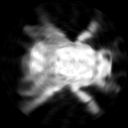}}

\subfigure[40 view slow-MBIR (NRMSE=0.5220) ]{\includegraphics[clip,width=0.19\textwidth]
{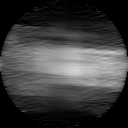}}
\subfigure[40 view fast-MBIR (NRMSE=0.1462) ]{\includegraphics[clip,width=0.19\textwidth]
{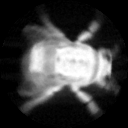}}
\subfigure[40 view fast-IFBP (NRMSE=0.1207) ]{\includegraphics[clip,width=0.19\textwidth]
 {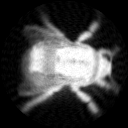}}
\subfigure[40 view fast-CodEx (NRMSE=0.1037)]{\includegraphics[clip,width=0.19\textwidth]
{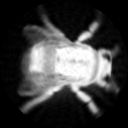}}
\subfigure[40 view coded-CodEx (NRMSE=0.0989)
]{\includegraphics[clip,width=0.19\textwidth]
{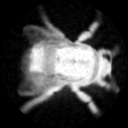}}
\caption{
Comparison of CodEx with other approaches for a simulated short duration scan.
The phantom is shown in Figure~\ref{fig:simresults_phantom_dynamic}.
For the 20 view case, the combination of fast scanning with a boxcar code and CodEx reconstruction performs best, and for the case of 40 views, the combination of fast scanning with a Raskar code and CodEx reconstruction performs best.
}
\label{fig:simresults_dynamic}
\end{figure*}

\begin{figure}
    \centering
    \includegraphics[clip,width=0.2\textwidth]{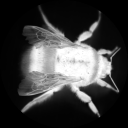}
    \caption{Phantom used for the experimental results in Figure~\ref{fig:simresults_dynamic}.}
    \label{fig:simresults_phantom_dynamic}
\end{figure}

\subsection{Binned Experimental Data}

\begin{figure*}[!ht]
\subfigure[Pseudo phantom]{\includegraphics[clip,width=0.19\textwidth]
{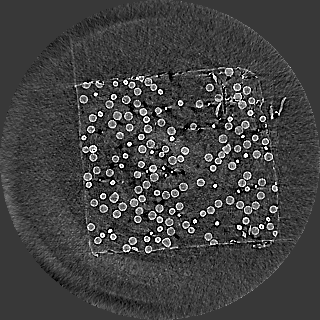}}
\subfigure[MBIR-boxcar]{\includegraphics[clip,width=0.19\textwidth]
{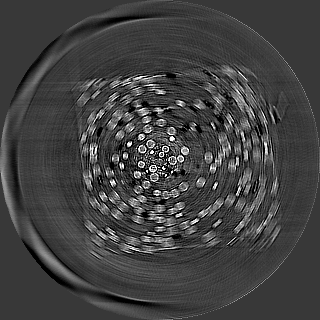}}
\subfigure[MBIR-Raskar]{\includegraphics[clip,width=0.19\textwidth]
{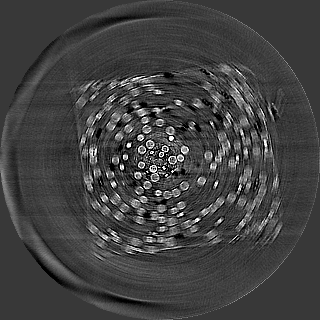}}
\subfigure[CodEx-boxcar]{\includegraphics[clip,width=0.19\textwidth]
{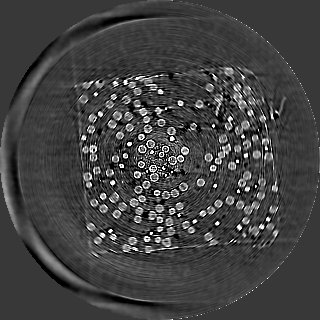}}
\subfigure[CodEx-Raskar]{\includegraphics[clip,width=0.19\textwidth]
{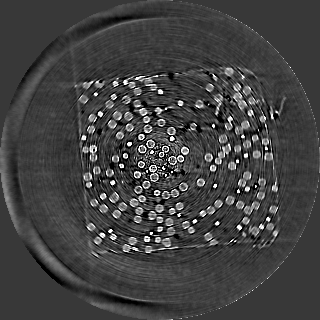}}
\caption{Comparison of reconstruction quality for binned physical data.
Experimental parameters are summarized in table~\ref{table:binsphere_setup}.
For the Raskar and boxcar codes, direct reconstruction produces a blurred image.
In contrast, CodEx is able to reconstruct the image without suffering from severe blurring.
}
\label{fig:binresults_spheres}
\end{figure*}

\begin{figure}[!ht]
\centering     
\subfigure[Primal Residual: RMSE($Ax^t$, $p^t$)]{\includegraphics[clip,width=0.24\textwidth]
{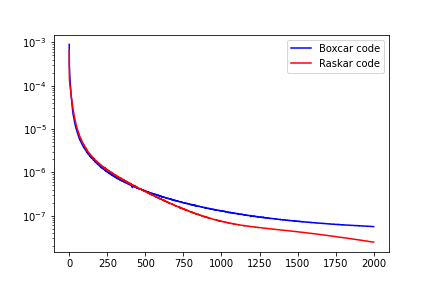}}
\subfigure[Dual Residual: RMSE($Ax^t$, $Ax^{t-1}$)]{\includegraphics[clip,width=0.24\textwidth]
{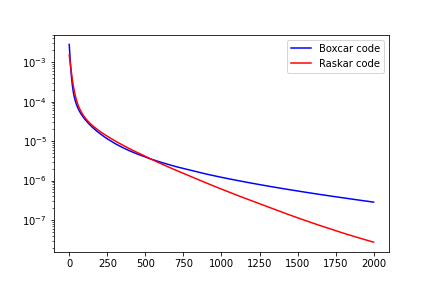}}
\caption{Primal and dual residual convergence plots for binned physical data.
Here $t$ refers to the ADMM iteration number.
The Raskar code leads to a slightly improved convergence than the boxcar code.
}
\label{fig:convplots_binsphere}
\end{figure}

In this section, we perform semi-simulated experiments by binning physical experimental data.
Binning physical data allows us to generate arbitrary coded view measurements by binning appropriate views, and we can also generate a ``pseudo phantom'' by using all the view data in a single high quality reconstruction.
The object in consideration contains borosilicate glass spheres of different sizes encased in a polypropylene matrix~\cite{singh2017varied}.
The experimental parameters are summarized in Table~\ref{table:binsphere_setup}.
In this case, we note that $K$ and $N_\theta$ have a gcd of 4, so we limit ourselves to $M_\theta =375$ views in order to avoid repetition of view-angles.

\begin{table}[!ht]
\centering{} 
\small
\begin{tabular}{r|l}
\toprule
Unique Micro-projection Angles: $N_{\theta}$  & 1,500 \\
Number of Views: $M_{\theta}$  & 375 \\
Code Length: $K$ & 52 \\
Blur Angle & 6.24$^\circ$ \\
Angular View Span & 6.48 rotations \\
Reconstruction Shape & 320$\times$320 \\
\bottomrule
\hline
\end{tabular}
\\
\vspace{1mm}
\caption{Parameters for binned physical experiment}
\label{table:binsphere_setup}
\end{table}

Using the original experimental measurements at $N_{\theta}$  distinct view-angles, we generate $M_{\theta}$ coded measurements as
\begin{align}
    y_i = -&\log\left\{ \sum_{k=0}^{K-1} \frac{c_k}{\bar{c}} \exp \left\{ -\tilde{y}_{\text{mod}(iK+k,N_{\theta})}  \right\} \right\} \ , \\ \notag
    & \text{  for } i=0,\cdots, M_{\theta}-1 \ ,
\end{align}
where $\tilde{y}_{\text{mod}(iK+k,N_{\theta})}$ is the vector of projection measurements at angle $\frac{\pi(iK+k)}{N_\theta}$ obtained from a physical experiment, $y_i$ is the vector of coded measurements at the $i^{\text{th}}$ view angle, $c = [c_0,\cdots,c_{K-1}]$ is the binary code of length $K$.

Figure~\ref{fig:binresults_spheres} shows a comparison of CodEx with the MBIR using the boxcar and Raskar codes. For the Raskar and boxcar codes, naive MBIR reconstruction produces a blurred image. In contrast, CodEx is able to reconstruct the image without suffering from severe blurring.

In Figure~\ref{fig:convplots_binsphere} we plot the primal residual, RMSE($Ax^t$, $p^t$) and dual residual, RMSE($Ax^t$, $Ax^{t-1}$) \cite{boyd2011distributed} at each ADMM iteration to illustrate the convergence. 
Here $t$ refers to the ADMM iteration number. The Raskar code leads to a slightly improved convergence than the boxcar code.

\section{Conclusion}

In this paper, we proposed CodEx, a novel method for coded exposure CT acquisition and reconstruction.
CodEx reconstruction models a) the linear blur that occurs during fast rotation; b) the nonlinear effects in transmission CT scanning; and c) the coding that can be temporally incorporated to improve resolution recovery.
Moreover, the CodEx algorithm can be implemented in a computationally efficient modular framework by using the ADMM algorithm.
Our experiments demonstrate that the CodEx method is more effective than traditional linear deblurring or MBIR reconstruction methods at improving reconstruction quality when using sparse view fly scanning of rapidly rotating objects. 
In addition, we demonstrate that temporal coding can be used to reduce blurring artifacts when scan view ranges over 180$^\circ$ are used.

\begin{appendix}
\label{sec:apendix}

\newtheorem{theorem}{Theorem}

\begin{theorem}\label{theorem_unique_angles}
All angles $\theta_i = i K \frac{\pi}{N_{\theta}} $ are unique (modulo $\pi$) if $0 \leq i \leq N_{\theta}-1$ and $\text{gcd}(K,N_{\theta})=1$.
\end{theorem}
\begin{proof}

Let us assume for the sake of contradiction that there are two integers $i$, $j$ such that $i \neq j$,  $0 \leq i,j \leq N_{\theta}-1$, and $\theta_i = \theta_j $ (modulo $\pi$).

Using the definition of $\theta_i$, this implies
\begin{equation}
\label{eq:th1_1}
    i K  = j K + c_1 N_{\theta} \ ,
\end{equation}
where $c_1$ is an integer constant.
Rearranging equation~\eqref{eq:th1_1}, we have
\begin{equation}
\label{eq:th1_2}
    (i-j) K = c_1 N_{\theta} \ .
\end{equation}
Now since the left-hand-side of equation~\eqref{eq:th1_2} is a multiple of K, so must be the right-hand-side.
However, since $\text{gcd}(K,N_{\theta})=1$, $c_1$ must be a multiple of K.
Let $c_1 = K c_2$, for some integer constant $c_2$.
Then, equation~\eqref{eq:th1_1} becomes
\begin{equation}
    i K  = j K + K c_2 N_{\theta} \ .
\end{equation}
Dividing by $K$ on both sides give
\begin{equation}
\label{eq:th1equiv}
    i = j + c_2 N_{\theta} \ .
\end{equation}
However, our initial assumption of $0 \leq i,j \leq N_{\theta}-1$ and $i \neq j$ is a direct contradiction to equation~\eqref{eq:th1equiv}.

\end{proof}

\begin{theorem}\label{theorem_N_theta}
If $K$, $m$, $n$ are integers such that $\text{gcd}(K,n)=1$ then $\text{gcd}(K,mK-n)=1$
\end{theorem}
\begin{proof}

For the sake of contradiction let us assume $\text{gcd}(K,mK-n)\neq1$
Thus we have $K = aq$ and $mK-n = bq$ for some integers $a$, $b$, $q$. 

Therefore, $n = mK-bq=q(ma-b)$.
Thus we have an integer $q$ that divides both $K$ and $n$ making $\text{gcd}(K,n)\neq1$ and leading to a contradiction.

\end{proof}

\end{appendix}

\bibliographystyle{IEEEtran}
\bibliography{ref.bib}

\end{document}